%% file: main.tex
\documentclass[11pt,a4paper]{article}
\usepackage{times,latexsym}
\usepackage{url}
\usepackage[T1]{fontenc}

    \usepackage[acceptedWithA]{tacl2021v1}
\usepackage[]{tacl2021v1}
\usepackage{xspace,mfirstuc,tabulary}

\input{macros}

\hypersetup{
    breaklinks=true,   
}

\newif\iftaclinstructions
\taclinstructionsfalse 
\iftaclinstructions
\renewcommand{\confidential}{}
\renewcommand{\anonsubtext}{(No author info supplied here, for consistency with
TACL-submission anonymization requirements)}
\newcommand{\instr}
\fi

\iftaclpubformat 

\else

\fi

\title{An Analysis of On-the-fly Determinization of Finite-state Automata}

\author{Ivan Baburin \\
  ETH Z{\"u}rich \\ 
  \texttt{ivan.baburin@inf.ethz.ch} \\\And
  Ryan Cotterell \\
  ETH Z{\"u}rich \\
  \texttt{ryan.cotterell@inf.ethz.ch} \\}

\date{}

\begin{document}

\maketitle
\begin{abstract}
  In this paper we establish an abstraction of on-the-fly determinization of finite-state automata using transition monoids and demonstrate how it can be applied to bound the asymptotics. We present algebraic and combinatorial properties that are sufficient for a polynomial state complexity of the deterministic automaton constructed on-the-fly. A special case of our findings is that automata with many non-deterministic transitions almost always admit a determinization of polynomial complexity. Furthermore, we extend our ideas to weighted finite-state automata.\looseness=-1

\end{abstract}

\section{Introduction}


One of the fundamental results in the theory of finite-state automata is the fact non-deterministic automata are equivalent in terms of computational power to deterministic ones. 
The classical conversion utilizes a power set construction, originally presented by \citet{RabinPowerset}, which results in an exponential blow-up of state number. Their upper bound was proven to be tight by \citet{moore1971bounds} who exhibited a simple automaton\footnote{This automaton is known as Moore's automaton.} (given in \cref{fig:Moore-NFA}) with $n$ states that requires at least $2^n$ states for its determinization.

Much research has be published to discuss the bounds on the state complexity of minimal deterministic automata for different types of subregular languages. \citet{BORDIHN20093209} presents a summary of many known results demonstrating that almost all classes of subregular languages will in the worst case require exponential state complexity. Other methods measure the amount of non-determinism in general automata and connect it with the state complexity, e.g., \citet{HromkovicTree} motivates the use of tree width as a measure of automaton complexity.\looseness=-1

We will take a different approach motivated by the practical application of finite-state techniques.
Instead, analyse the on-the-fly variant of the classic power set construction, originally presented by \citet{Mohri-on-the-fly}. Surprisingly, despite its numerous  practical applications as in \citet{allauzen2007openfst}, there has not been a concise analysis of its complexity presented in the literature. 
We will attempt to bridge the gap and present a simple abstraction to estimate its performance and present different criteria which would ensure its efficient (i.e., polynomial time) execution on various classes of automata. Although the criteria are restrictive, they can still be used as an intuition to estimate the state complexity in certain practical applications.

Our main idea lies in connecting the execution paradigm of on-the-fly construction to the size of transition monoid and formulating sufficient criteria which translate into polynomial state complexity of the output. 
Moreover, since on-the-fly algorithm produces a valid deterministic automaton our bounds additionally translate into upper bounds for minimal deterministic automata. In summary:\looseness=-1
\begin{figure}
    \centering
    \begin{tikzpicture}
    \node[state, initial] (q1) {$q_1$};
    \node[state, right of=q1] (q2) {$q_2$};
    \node[state, right of=q2] (q3) {$q_3$};
    \node[fill=white, right of=q3] (q?) {$\ldots$};
    \node[state, accepting, right of=q?] (qn) {$q_n$};
    \draw (q1) edge[loop above] node{$b$} (q1)
          (q1) edge[above] node{$a$} (q2)
          (q2) edge[above] node{$a, b$} (q3)
          (q3) edge[above] node{$a, b$} (q?)
          (q?) edge[above] node{$a, b$} (qn)
          (qn) edge[bend right=1.7cm, above] node{$a$} (q2)
          (qn) edge[bend left=1.5cm, below] node{$a$} (q1);
    \end{tikzpicture}
    \caption{Moore's non-deterministic automaton with $n$ states which for $n \geq 2$ requires at least $2^n$ states in its deterministic analogue}
    \label{fig:Moore-NFA}
\end{figure}
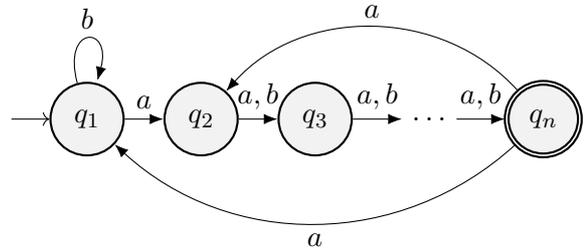

\begin{itemize}
    \item We formalize the connection between the on-the-fly algorithm and transition monoids, and show how it can be used to rederive known state complexity bounds for one-letter and commutative automata;
    \item We define strongly connected automata and establish state complexity bounds for them;
    \item We prove that almost all automata with large number of transitions will  only require a polynomial number of states;
    \item We extend our findings to weighted finite-state automata.
\end{itemize}




\section{Finite-state automata}

We start by briefly reciting core definitions for finite-state automata.

\begin{definition}
    A \defn{finite-state automaton (FSA)} $\automaton$ is a quintuple $\fsatuple$ where $\alphabet$ is an alphabet, $\states$ is a finite set of states, $\initial \subseteq \states$ is a set of starting states, $\final \subseteq \states$ is a set of accepting states and $\trans \subseteq \states \times \alphabet \cup \set{\eps} \times \states$ a finite multi-set of transitions.\looseness=-1
\end{definition}

The symbol $\eps$ represents an empty symbol, and thus a transition with label $\eps$ (or in short $\eps$-transition) can always be performed without the need for any additional input. 

\begin{definition}\label{def:path}
    A \defn{path} $\apath \in \kleene{\trans}$ is a sequence of consecutive transitions of the form $\edgenoweight{\stateq_0}{\bullet}{\stateq_1} \edgenoweight{}{\bullet}{\ldots}  \edgenoweight{}{\bullet}{\stateq_m}$ where $\bullet$ is a placeholder for the transition label. We will refer to the concatenation of symbols along the path as its \defn{yield}.
\end{definition}

We say that a word $\str \in \kleene{\alphabet}$ is \defn{recognized} by the FSA $\automaton$ if there exists a path from some starting state $\stateq \in \initial$ to some final state $\stateq' \in \final$ with yield $\str$. Moreover, we denote with $\languagefsa{\automaton}$ the \defn{language} (set of all words) recognized by $\automaton$. Two FSAs $\automaton$ and $\automaton'$ are called \defn{equivalent} if $\languagefsa{\automaton} = \languagefsa{\automaton'}$. 

\begin{definition}\label{def:deterministic}
    FSA $\detautomaton$ is called \defn{deterministic} if and only if it does not contain any $\eps$ transitions, the starting state is unique, i.e., $|\initial| = 1$, and for every $(\stateq, \syma) \in \states \times \alphabet$ there is at most one $\stateq' \in \states$ such that $(\stateq, \syma, \stateq') \in \trans$, thus we have at most one unique labeled transition from every state.
\end{definition}

Transforming a non-deterministic automaton into an equivalent deterministic one is a procedure we refer to as \defn{determinization}. 
Although the complexity of the asymptotically optimal algorithm lies in $\mathrm{EXPSPACE}$, one often uses the on-the-fly version of the classical power set construction, which for many automata is substantially more efficient.

We present the on-the-fly construction for completeness (\cref{alg:onthefly}) and define a notion of a \defn{power state} $\powerstate \subseteq \states$. Given a non-deterministic automaton $\automaton$ its deterministic counterpart will have a state space be a subset of a power set construction $\detstates \subseteq \powerset{\states}$. On-the-fly algorithm will start in the power state $\initpowerstate \coloneqq \{ \stateq \ | \ \stateq \in \initial \}$ and explore all possible power states that can be reached following the labeled transitions $\trans$.

\begin{algorithm}[t]
\caption{\textsc{On-the-fly}}\label{alg:onthefly}

\begin{algorithmic}
\Ensure $\automaton$ a FSA with no $\eps$-transitions
\State $\detautomaton \gets (\alphabet, \detstates, \powerstate_{\initial}, \detfinal, \dettrans)$
\State $\mathsf{stack} \gets \powerstate_\initial$
\State $\detstates \gets \{ \powerstate_\initial \}$
\While{$|\mathsf{stack}| > 0$}
    \State \textbf{pop} $\powerstate$ from the $\mathsf{stack}$
    \ForAll{$\syma \in \alphabet $}
        \State $\powerstate' \gets \{\stateq' \ | \ (\stateq, \syma, \stateq') \subseteq \trans, \ \stateq \in \powerstate \}$
        \State $\dettrans \gets \dettrans \cup \{\edgenoweight{\powerstate}{\syma}{\powerstate'} \}$
        \If{$\powerstate' \notin \detstates$}
            \State $\detstates \gets \detstates \cup \{ \powerstate' \}$
            \State \textbf{push} $\powerstate'$ on the $\mathsf{stack}$
        \EndIf
    \EndFor
\EndWhile
\State $\detfinal \gets \{ \powerstate \in \detstates \ | \ \powerstate \cap \final \neq \varnothing \}$
\State \textbf{return} $\gets \detautomaton$
\end{algorithmic}
\end{algorithm}

Correctness of \cref{alg:onthefly} follows from the power set construction, and we have at most $2^{\nstates}$ iterations of the while loop because every power state will appear on the stack at most once, thus the  worst case execution time is in $\bigO{\nsymbols \nstates 2^{\nstates}}$.
Moreover we can assume without loss of generality that our FSA $\automaton$ is $\eps$-free, since there exist a handful of asymptotically efficient procedures for computing an equivalent FSA containing no $\eps$-transitions (even for weighted automata), for example as presented by \citet{allauzen2007openfst}.


\section{Characterization using transition monoids}\label{sec:transition_monoid}

We start with a swift introduction to algebraic automata theory, as presented by \citet{pin1997syntactic}.

\begin{definition}
    Consider a FSA $\automaton = \fsatuple$ and for each $\syma \in \alphabet$ define a binary relation ${T}^{(\syma)}$ over $\states$ with $T^{(\syma)}(\stateq) := \trans(\stateq, \syma)$ for all $\stateq \in \states$. A \defn{transition monoid} $\transmonoid{\automaton}$ is a binary relation monoid over $\states$ generated by $ \{ T^{(\syma)} \ | \ \syma \in \alphabet \}$ and closed under relation composition operator $\circ$ with the identity relation $\mathit{id}$.\looseness=-1
\end{definition}

It is often more convenient to think of $\transmonoid{\automaton}$ as a monoid of $\nstates \times \nstates$ matrices over Boolean semifield $\Boolean$ closed under Boolean matrix multiplication and the identity element $\sI$, such that:
\begin{equation}
    T^{(\syma)}_{i, j} = 1 \Leftrightarrow (\stateq_i, \stateq_j) \in T^{(\syma)}
\end{equation}

Notice that the monoid $\transmonoid{\automaton}$ is per construction closely related to the regular language $\languagefsa{\automaton}$. This property can be further formalized algebraicly.

\begin{definition}
    A language $\sL \subseteq \kleene{\alphabet}$ is \defn{recognized by a monoid} $\monoid$ of binary relations over $\states$ if there exists a surjective morphism $\morphismdef : \kleene{\alphabet} \rightarrow \monoid$ and an \defn{accepting subset} $\monoid_{\final} \subseteq \monoid$ such that $\sL = \morphismdef^{-1}(\monoid_{\final})$.\looseness=-1
\end{definition}

From the definition above it is immediately visible that $\transmonoid{\automaton}$ recognizes the language $\languagefsa{\automaton}$. Indeed, if we view $\initial$ and $\final$ as Boolean vectors we can define the accepting subset as 
\begin{equation}
    \monoidS_{\final} := \left\{ M \in \transmonoid{\automaton} \ | \ \initial^\top M \final \neq 0 \right\}
\end{equation}

\noindent and a morphism $\morphismdef$ for an arbitrary word $\str = \syma_1\syma_2\syma_3 \ldots \syma_m \in \kleene{\alphabet}$
\begin{equation}
    \morphism{\str} = \sI \cdot \prod^m_{i=1} T^{(\syma_i)} \in \transmonoid{\automaton}
\end{equation}

\noindent with $\morphism{\eps} = \sI$. As a consequence, the inverse $\morphismdef^{-1}(\monoidS_{\final})$ only contains words $\str \in \kleene{\alphabet}$ that are a yield of at least one path from $\initial$ to $\final$ in $\automaton$.

Due to equivalence between FSAs and their transition monoids there are many interconnected results. For instance, the smallest monoid, known as \defn{syntactic monoid}, recognizing some regular language $\sL$ is isomorphic to the transition monoid of the minimal deterministic FSA accepting $\sL$.

We consider the execution paradigm of \cref{alg:onthefly} through the lens of transition monoids. Notice, in every step the algorithm effectively multiplies the top element on the stack with transition matrices, and terminates as soon as no new power states can be produced.

\begin{observation}\label{obs:on-the-fly}
    The set of states $\detstates$ in the automaton $\detautomaton$ is the image space of $\initial$ under the transition monoid $\transmonoid{\automaton}$ i. e.
    \begin{equation}
        \detstates \!=\! \left\{ \powerstate \subseteq \states \ | \ \powerstate = \initial^\top M, \ M \in \transmonoid{\automaton} \right\} 
    \end{equation}
    
\end{observation}

This allows us to formulate a simple, but nevertheless interesting, connection between the state complexity of \cref{alg:onthefly} and the size of the transition monoid $\transmonoid{\automaton}$.\looseness=-1

\begin{theorem}\label{thm:monoid}
    Given a FSA $\automaton$ the equivalent deterministic automaton $\detautomaton$ computed by means of \cref{alg:onthefly} will satisfy

    \begin{equation}
        |\detstates| \leq |\transmonoid{\automaton}|
    \end{equation}
    
    \noindent and moreover its construction will require at most $\bigO{\nsymbols \nstates |\transmonoid{\automaton}|}$ steps.
\end{theorem}

\begin{proof}
    Following \cref{obs:on-the-fly} we conclude that since every element of $\detstates$ is the image of at least one relation in $\transmonoid{\automaton}$ we have $|\detstates| \leq | \transmonoid{\automaton} |$. Since the number of \texttt{while}-loop iterations in \cref{alg:onthefly} corresponds to the number of states, and in each iteration we consider at most $\nsymbols \nstates$ transitions the second claim follows.
\end{proof}

We remark that in general transition monoids tend to be very large, sometimes much larger than the actual number of states in the FSA. For $n$-state deterministic FSAs over a binary alphabet the syntactic monoid recognizing the same language can have the size of up to $n^n \left(1 - \frac{2}{\sqrt{n}} \right)$. A concise discussion of different exponential bounds can be found in \citet{holzer2004deterministic}.

As a result, to achieve polynomial bounds for \cref{thm:monoid} we will require additional constraints that we can impose on the automaton $\automaton$. In the latter sections we will present different approaches for different classes of automata.

\section{One-letter automata}\label{sec:one-letter}

In this section we consider a relatively simple class of automata whose alphabet consists only of a single symbol.
Such automata are often called \defn{one-letter} automata.
Let $\automaton_1 = (\{ \syma \}, \states_1, \trans_1, \initial_1, \final_1)$ be an arbitrary one-letter FSA. 
It is immediately clear that the transition monoid $\transmonoid{\automaton_1}$ of a one-letter FSA is generated entirely by a Boolean matrix $A$ defined as:\looseness=-1
\begin{equation}
    A_{i, j} = \left\{ \begin{array}{rl}
                    1 & \textit{if} \quad (\stateq_i, \syma, \stateq_j) \in \trans_1\\ 
                    0 & \textit{else} 
                    \end{array}\right.
\end{equation}

We follow the approach originally presented by \citet{Markowsky_1976} and analyze matrix semigroups generated by a single Boolean matrix.

\begin{definition}\label{def:index_period}
    The \defn{index} of a Boolean matrix $B$, denoted as $\indexM{B}$, is the least positive integer $k$ such that $B$ satisfies:
    \begin{equation}\label{eq:def-index-period}
        B^{k+d} = B^k \quad \text{for some $d \in \N_{>0}$}
    \end{equation}
    \noindent Similarly, the least possible choice of $d$ in the equation above given $k = \indexM{B}$ is the \defn{period} of $B$, denoted as $\periodM{B}$.
\end{definition}

Since there are exactly $2^{n^2}$ possible $n \times n$ Boolean matrices it immediately follows that every Boolean matrix has a finite index and period. 
Note that the size of the transition monoid satisfies
\begin{equation}\label{eq:one-letter-transition-monoid}
    |\transmonoid{\automaton_1}| \leq \indexM{A} + \periodM{A}
\end{equation}
because it contains all elements from the semigroup generated by $A$ and the identity element $\sI$.
Thus, to ensure a polynomial bound on the size of the transition monoid we require a polynomial bound both for the period and the index of $A$.

\begin{proposition}[\citet{Markowsky_1976}]\label{prop:bound-index}
    For an arbitrary $n \times n$ Boolean matrix $B$ we have:
    \begin{equation}
        \indexM{B} \leq n^2 - 2n + 2
    \end{equation}
\end{proposition}
\cref{prop:bound-index} implies that the index of an arbitrary matrix will be relatively small, thus the main complexity always comes from the period.

\citet{denes1983automata} investigated lower bounds connected to periods of Boolean matrices and showed that, in general, they can be exponentially large. Moreover, they demonstrated how their bounds  translate directly into the number of states in a minimal equivalent deterministic one-letter automata, thus confirming that certain one-letter automata will always have exponential state complexity.
To achieve a polynomial bound on the size of $\transmonoid{\automaton_1}$ we restrict the discussion to a constrained class of transition matrices.

\begin{definition}\label{def:irreducible}
    A Boolean matrix $B$ is called \defn{irreducible} if and only if there does not exists a permutation matrix $P$ such that $PBP^\top$ is a block lower triangular matrix.
\end{definition}

This condition is equivalent to the directed graph defined by adjacency matrix $A$, known as the \defn{precedence graph} $\graph(A)$, being strongly connected. 
For more details on this observation we refer to \citet{brualdi1991combinatorial}.

\begin{proposition}\label{prop:bound-period}
    Given an arbitrary irreducible $n \times n$ Boolean matrix $R$ it holds:
    \begin{equation}
        \periodM{R} \leq n
    \end{equation}
\end{proposition}

\begin{proof}
    The period of an arbitrary Boolean matrix $R$ was shown by \citet{DeSDeM:98-20} to be equal to the least common multiple of cyclicities of all maximal strongly connected components of $\graph(R)$. Since $R$ is irreducible, $\graph(R)$ is strongly connected, thus its cyclicity can be trivially bounded by $n$ (length of the largest cycle) from above.\looseness=-1
\end{proof}

This allows us to establish a simple polynomial bound for FSAs with irreducible transition matrices.\looseness=-1
\begin{lemma}\label{lem:one-letter}
    An equivalent one-letter deterministic automata computed by applying \cref{alg:onthefly} to automaton $\automaton_1$ with an irreducible transition matrix $A$ and $|\states_1| = n$ will have at most $n^2 - n + 2$ states.\looseness=-1
\end{lemma}

\begin{proof}
    By inserting \cref{prop:bound-period} and \cref{prop:bound-index} into \cref{eq:one-letter-transition-monoid} we get the desired bound. The claim follows by \cref{thm:monoid}.
\end{proof}

We emphasize that while irreducibility of $A$ is sufficient for polynomial execution time of \cref{alg:onthefly}, it is by no means necessary. A more involved bound on the period of $A$ and its relation to cyclicity of $\graph(A)$ can be found in \citet{DeSDeM:98-20}.

\section{Commutative automata}\label{sec:commuting_automata}

We investigate how many properties from one-letter automata carry over to the general case. We remark that just requiring irreducibility for all transition matrices in $\automaton$ no longer suffices. Indeed, the product of irreducible matrices is no longer irreducible, thus we can easily create interim matrices with large period and blow-up the size of the transition monoid.

However for some cases $|\transmonoid{\automaton}|$ can still be bounded. Consider $\sT := \{ T^{(\syma)} \ | \ \syma \in \alphabet \}$ the set of all transition matrices of $\automaton$. Assume that each matrix satisfies:
\begin{equation}\label{eq:product-transition}
    \quad T^{(\syma_i)} = {\left(T^{(\syma_1)}\right)}^{k_i} \quad \text{for some }  k_i \in \N
\end{equation}

In this case we can just apply the bound from \cref{eq:one-letter-transition-monoid} on $T^{(\syma_1)}$ directly, since additional letters don't affect the size of $\transmonoid{\automaton}$.

One can extend the argument to a more general class of automata. We start with an assumption that all transition matrices in $\sT$ commute; this is a direct consequence of \cref{eq:product-transition}. 
Then the order of multiplication between the elements of $\sT$ is not relevant, thus one can show a tighter bound on the size of $\transmonoid{\automaton}$.

\begin{theorem}\label{thm:commutative}
    Given an automaton $\automaton$ let for all $i$ the size of the transition monoid generated by $T_i \in \sT$ be bounded by $f_i(n)$. Further assume that all matrices in $\sT$ commute with respect to Boolean matrix multiplication. Then:
    \begin{equation}
        |\transmonoid{\automaton}| \leq \prod_{i=1}^{\nsymbols} f_i(n)
    \end{equation}
     which is also the bound on the state complexity of \cref{alg:onthefly} applied on $\automaton$.
\end{theorem}

\begin{proof}
    Consider an arbitrary element $M \in \transmonoid{\automaton}$. Using commutativity of $\sT$ we can write $M$ as a product of transition matrices as following:
    \begin{equation}
        M = \prod_{i = 1}^{\nsymbols} {\left(T^{(\syma_i)}\right)}^{k_i} \quad \text{for some $k_i \in \N$}
    \end{equation}
    Since all matrices $T^{(\syma_i)}$ generate a monoids of bounded size we can have at most $f_i(n)$ distinct values of ${\left(T^{(\syma_i)}\right)}^k$, thus every element from $\transmonoid{\automaton}$ can be encoded by at least one element from $\bigtimes_{i = 1}^{\nsymbols} \bigl\{0, \ldots, f_i(n)-1 \bigr\}$. 
\end{proof}

\begin{corollary}\label{corr:commutative}
    Assuming the preconditions from \cref{thm:commutative} and irreducibility of all transition matrices in $\sT$ the \cref{alg:onthefly} applied to $\automaton$ would have a state complexity of at most $n^{2\nsymbols}$.
\end{corollary}

\begin{proof}
    We notice by \cref{lem:one-letter} that for an arbitrary $i$ we have $f(i) \leq n^2 - n + 2 \leq n^2$ (given $n \geq 2$). The claim follows.
\end{proof}

Notice that the conditions from \cref{thm:commutative} essentially enforce that the transition monoid $\transmonoid{\automaton}$ is commutative. Automata with commutative transition monoids correspond to commutative languages, which have other interesting algebraic properties. 
\citet{hoffmann2019commutative} presents a survey of different properties of commutative languages and demonstrates that a ``shuffle'' of two regular commutative languages can be performed in polynomial time, something that has exponential complexity for general languages. They also derive a similar bound to \cref{thm:commutative} for the state complexity of the minimal deterministic automaton via underlying one-letter languages. However they utilize a different, non-algebraic approach without establishing the link to on-the-fly algorithm.\looseness=-1

\section{Strongly connected automata}\label{sec:strongly-connected}

In this section we take a different approach and extend the notion of irreducibility to bound the state complexity of automata by introducing additional constraints on their transitions. 
\begin{definition}
    An $n \times n$ Boolean matrix $B$ is called $r$-\defn{indecomposable} if and only if there do not exist permutation matrices $P$ and $Q$ such that $PBQ$ has a top right $\0$-block of size $s \times t$ with $s + t + r > n$:
    \begin{equation}
        PBQ \neq 
        \begin{bmatrix}
            \ast & \0 \\
            \ast & \ast
        \end{bmatrix}
        \quad \forall P, Q \in \permgroup
    \end{equation}
Moreover we call $B$ $r$-\defn{irreducible} if the precedence graph $\graph(B)$ is $r$-\defn{connected}, i.e., one can remove arbitrary $r-1$ vertices from $\graph(B)$ and it will remain strongly connected.\looseness=-1
\end{definition}
Indecomposable matrices are related to standard irreducible matrices. 
Specifically, every irreducible matrix is $0$-indecomposable and every $1$-indecomposable matrix is irreducible. 
In general one can shown that both of these notions are very similar, i.e., by adding some minor conditions (positivity of the main diagonal) one can derive an equivalence.
\begin{lemma}[\citet{YouIndecomposable}]\label{lem:r-strongly-connected}
    Suppose a Boolean matrix $B$ has only ones along the main diagonal. 
    Then, for $r > 0$, it is $r$-indecomposable if and only if its precedence graph $\graph(B)$ is $r$-connected (i.e., $B$ is $r$-irreducible).
\end{lemma}
\cref{lem:r-strongly-connected} allows us to verify whether a transition matrix is $r$-indecomposable in polynomial time. Indeed, notice that indecomposability of a matrix is invariant under permutations. Thus, assuming we permuted the matrix $A$ such that the main diagonal consists entirely of ones we only have to check whether the underlying graph $\graph(A)$ is $r$-connected. 

\begin{observation}
    Verifying graph connectivity is a well-known instance of maximum flow problem. A directed graph $\graph = (V, E)$ is $r$-connected if and only if for every pair of vertices $u, v \in V$ with $(u, v) \notin E$ there exist at least $r$ vertex disjoint directed paths from $u$ to $v$.
\end{observation}

Due to the almost linear time algorithm for calculating maximum flow in a graph by \citet{ChenMaxFlow} the $r$-connectivity of $\graph(A)$ can be verified in $\bigO{n^2 \cdot |A|^{1 + o(1)}}$ steps where $|A|$ refers to the number of non-zero entries in $A$.

\begin{proposition}[\citet{kim1978two}] \label{prop:indecomposable}
    An $n \times n$ Boolean matrix $B$ is $r$-indecomposable if and only if for any non-zero Boolean (row) vector $v$ it holds:
    \begin{equation}
        |vB| \geq \min{\{n, \ | v | + r \}}
    \end{equation}
\end{proposition}

Next we use the notion of $r$-indecomposabilty to bound the size of the transition monoid.

\begin{lemma}\label{lem:indecomposable}
    Consider an automaton $\automaton$ with the set of transition matrices $\sT$ as defined previously, such that all transition matrices in $\sT$ are $r$-indecomposable with $r > 0$. Then $\transmonoid{\automaton}$ will contain at most $\frac{\nsymbols^{\lceil (n - 1)/r \rceil}}{\nsymbols-1} + 1$ elements.
\end{lemma}

\begin{proof}
    Consider a product of the identity matrix $\sI$ with arbitrary $\left\lceil \frac{n-1}{r} \right\rceil$ matrices from $\sT$. Due to \cref{prop:indecomposable} each column of the resulting product will have at least $ 1 + \left\lceil \frac{n-1}{r} \right\rceil\cdot r = n$ entries, thus we will get a complete matrix \texttt{1}. Thus all non-\texttt{1} elements of $\transmonoid{\automaton}$ will be the products of at most $\left\lceil \frac{n - 1}{r} \right\rceil - 1$ transition matrices. If we sum over all possible lengths $i$ we get the desired bound:
    \begin{equation}
        |\transmonoid{\automaton}| - 1 \leq \!\!\! \!\!\!\sum_{i = 0}^{\lceil (n-1)/r \rceil - 1} \!\!\!\nsymbols^i \leq \frac{\nsymbols^{\lceil (n-1) / r \rceil}}{\nsymbols - 1}
    \end{equation}
    Counting the \texttt{1}-matrix concludes the proof.
\end{proof}

\begin{corollary}\label{corr:indecomposable}
    Following the definitions from \cref{lem:indecomposable} if all transition matrices are at least $\left\lceil \frac{n}{k\log{n}} \right\rceil $-indecomposable the state complexity of \cref{alg:onthefly} will be at most $2n^{k \log{\nsymbols}}$.
\end{corollary}

We notice that contrary to \cref{corr:commutative} the state complexity only has $\log{\nsymbols}$ in the power, which is a big improvement over $\nsymbols$. \citet{YouIndecomposable} lists some examples of $r$-connected graphs for large values of $r$.

\section{Automata with many transitions}\label{sec:automata-many-transitions}

Another interesting class of automata one would expect to have bounded state complexity are automata with many non-deterministic transitions. For example in the extreme case of an automaton with a complete transition relation already $2$ states are sufficient for its deterministic counterpart (the starting state and the complete power state). We will formulate the bounds for such automata by relating them to strongly connected automata from \cref{sec:strongly-connected}.
To achieve a polynomial state complexity for \cref{alg:onthefly} we want to apply \cref{corr:indecomposable}, however it does require a very strong indecomposability for all transition matrices. 

We demonstrate that almost all Boolean matrices with sufficiently many non-zero entries will be strongly indecomposable, thus almost all automata with sufficiently many transitions will have a polynomial state complexity.\looseness=-1

\begin{theorem}\label{thm:indecomposable}
    If we consider automaton $\automaton$ where we pick individual transition matrices in $\sT$ uniformly at random (however not necessarily independently from each other) from the set of $n \times n$ Boolean matrices with $n^2 / d$  non-zero entries for some constant $d$ with high probability it holds:
    \begin{equation}
        |\transmonoid{\automaton}| \leq 2\nsymbols^{d \log{n} + o(1)}
    \end{equation}
   Thus with high probability \cref{alg:onthefly} will have a polynomial state complexity of at most $\bigO{n^{d \log{\nsymbols}}}$.
\end{theorem}

\begin{proof}
    \citet{kim1978two} showed that an $n \times n$ Boolean matrix picked uniformly at random from the set of matrices with exactly $(1 + \epsilon + r)n \log{n}$ non-zero entries will be $r$-indecomposable with probability $1 - o(1)$ for an arbitrary $\epsilon > 0$. Notice, if we sample from the matrices with at least as many non-zero entries the bound will also hold. Using the union bound the probability that all $\nsymbols$ transition matrices will be $r$-indecomposable is thus also $1 - o(1)$. By \cref{lem:indecomposable} we can substitute $r = \frac{n}{d \log{n}} - (1 + \epsilon)$ and bound:
    \begin{equation}
        |\transmonoid{\automaton}| \leq {2{\nsymbols}^{(n-1) /r}} \leq 2\nsymbols^{d \log{n} + o(1)}
    \end{equation}
    with high probability. The rest follows from \cref{thm:monoid}.
\end{proof}
Notice that the condition for $r$-indecomposability from \cref{thm:indecomposable} to require at least $(1 + \epsilon + r) n\log{n}$ non-zero entries is logarithmically tight due to every $r$-indecomposable matrix having at least $(1+r)n$ non-zero entries.\looseness=-1
We can also characterize the language properties of such automata, and they are relatively simple --- if you select a sufficiently long word $\str \in \kleene{\alphabet}$ it will almost always be accepted. 

It is important to remember that not all such automata will have low state complexity, for example one can easily add many ``useless'' states with a lot of redundant transitions to Moore's automaton such that the underlying language is not affected.

\begin{remark}
    If we allow transition matrices to be picked independently, already $\bigOmega{\sqrt{n^3 \log{n}}}$ transitions being sufficient for the constant state complexity, originally derived by \citet{kim1978two} (for binary semigroups). Looking only at uncorrelated transition matrices is a big (and impractical) simplification, thus the extended bounds serve mostly mathematical curiosity.
\end{remark}

\section{Alternative approaches}\label{sec:alternative-approaches}

In this section we consider some criteria which are not connected to transition monoids. The motivation behind is the simple fact that our most general bound as presented in \cref{thm:monoid} is not even asymptotically tight when applied to deterministic automata. Indeed, as mentioned earlier for deterministic automata with $\nstates = n$ states and binary alphabet the size of the transition monoid can be close to $n^n \left(1 - \frac{2}{\sqrt{n}} \right)$. However it trivially holds that the on-the-fly algorithm will require at most $n$ states on an arbitrary deterministic automaton.

The observed discrepancy comes from the absence of information about the starting states of $\automaton$ in its transition monoid $\transmonoid{\automaton}$. Thus, even if  $|\transmonoid{\automaton}|$ is exponentially large (see \cref{obs:on-the-fly}), the number of states $\detstates$  can be much smaller, most notably for deterministic automata it will be $\leq n$. Essentially $|\transmonoid{\automaton}|$ acts as a bound for the worst possible choice of the initial set $\initial$, i.e., resulting in the biggest blow-up.
To mitigate this issue we consider some alternative approaches that were proposed to measure non-determinism and analyse how they relate to on-the-fly algorithm.

\begin{definition}
Let $\automaton = \fsatuple$ be a FSA and $\str \in \kleene{\alphabet}$. The \defn{tree width} of $\automaton$ on $\str$, denoted as $\tau_{\automaton}(\str)$, is the number of different paths with yield $\str$ starting in some initial state $\stateq \in \initial$. The tree width of $\automaton$ is defined as:
\begin{equation}
    \tw{\automaton} = \max{\left\{ \tau_{\automaton}(\str) \ | \ \str \in \kleene{\alphabet} \right\}}
\end{equation}
We say that $\automaton$ has finite tree width if $\tw{\automaton}$ is finite.\looseness=-1
\end{definition}

It is easy to see that the tree width of an arbitrary deterministic automaton is exactly $1$ since we have a unique starting state and a unique transition for each symbol. Moreover, tree width gives a natural upper bound for the state complexity of \cref{alg:onthefly}.

\begin{lemma}\label{lem:bounded-width}
    Given a FSA $\automaton$ with $\tw{\automaton} = k$ and $k \leq n-1$ the on-the-fly algorithm will require at most $\frac{n^k}{(k-1)!} + 1$ states.
\end{lemma}

\begin{proof}
    Since for an arbitrary word $\str$ the width of the computation tree $\tau_{\automaton}(\str)$ is bounded by $k$, at no point in the execution  will there be a power state $\powerstate \in \detstates$ with more than $k$ entries. Thus the total number of reachable power states can be bounded by summing up all possible combinations:
    \begin{equation}
        \sum_{i = 0}^{k} \binom{n}{i} \leq \sum_{i = 0}^{k} \frac{n^i}{i!} \leq \sum_{i = 0}^{k} \frac{n^k}{k!} = \frac{n^k}{(k-1)!} + 1
    \end{equation}
\end{proof}
\citet{PalioudakisFiniteTW} have shown the optimality of the state complexity derived in \cref{lem:bounded-width} by constructing a family of automata $\automaton_{n, k}$ with $n$ states and $\tw{\automaton_{n, k}} = k$ such that the minimal equivalent deterministic automaton has a state complexity of $1 + \sum_{i = 1}^k \binom{n-1}{i}$. Moreover, they give a complete characterization of all automata of bounded tree width.

\begin{lemma}[\citet{palioudakis2012state}]\label{lem:no-cycle}
    An FSA $\automaton = \fsatuple$ has finite tree width if and only if no directed cycle in $\trans$ contains a non-deterministic transition.
\end{lemma}

From \cref{lem:no-cycle} we can directly infer that automata with irreducible (and indecomposable) transition matrices we considered in \cref{sec:one-letter}, \cref{sec:commuting_automata} and \cref{sec:strongly-connected} will in general not have a bounded tree width. This indirectly confirms our previous observation that there is no direct connection between the tree width $\tw{\automaton}$ and size of transition monoid $\transmonoid{\automaton}$. 

\begin{theorem}[\citet{HROMKOVICthm}]\label{thm:Hromokovic}
    For an arbitrary FSA $\automaton$ with $n$ states the tree width $\tw{\automaton}$ is either bounded by a constant or between linear and superlinear, or otherwise $2^{\Theta(n)}$.
\end{theorem}

\cref{thm:Hromokovic} demonstrates that no quasipolynomial state complexity bounds can be established using \cref{lem:bounded-width}. Indeed, by applying \cref{lem:bounded-width} we get a state complexity bound of the form $\bigO{n^{\tw{\automaton}}}$, thus for a constant tree width we achieve polynomial complexity, and for anything linear or above an exponential state complexity (effectively producing a ``blow-up'' comparable to the naive power set construction).

\section{Weighted finite-state automata}\label{sec:appendix-weighted-automata}

Another useful property of the algebraic approach is the possibility of extending the state complexity bounds to general weighted finite-state automata. Despite the fact that contrary to classical automata they can't always be determinized on-the-fly, as shown by \citet{Mohri-on-the-fly}, a slight modification of our core result from \cref{thm:monoid} still holds.

\begin{definition}\label{def:semifield}
    An algebra $\langle \K, \oplus, \otimes, \0, \1 \rangle $ is called a \defn{(commutative) semifield} if and only if it satisfies the following properties:
    \begin{enumerate}
        \item $\langle \K, \otimes, \1 \rangle$ is a commutative group
        \item $\langle \K, \oplus, \0 \rangle$ is a commutative monoid
        \item $\otimes$ with $\0$ annihilates $\K$
        \item $\otimes$ distributes over $\oplus$
    \end{enumerate}
    Moreover we call the semifield \defn{zero-sum-free} if and only if for all $x, y \in \K$ it holds:\looseness=-1
    \begin{equation}
        x \oplus y = \0 \Rightarrow x = y = \0
    \end{equation}
\end{definition}

A couple prominent examples of zero-sum-free semifields are the Boolean semifield
\begin{equation}
    \Boolean \coloneqq \langle \{ 0, 1 \}, \vee, \wedge, 0, 1 \rangle
\end{equation}
which contains only two elements and the tropical semifield
\begin{equation}
    \Tropical \coloneqq \langle \R \cup \{ \infty \}, \min, +, \infty, 0 \rangle
\end{equation} 
which contains infinitely many elements. With the algebra out of the way we can present automata over general semifields. 

\begin{definition}
    A \defn{weighted finite-state automaton (WFSA)} $\wfsa$ over a (commutative) semifield $\langle \K, \oplus, \otimes, \0, \1 \rangle $ is a quintuple $\wfsatuple$ where each transition in $\trans \subseteq \states \times \alphabet \times \K \times \states$ has a weight over $\K$ and $\initf, \finalf$ are initial and final weight functions of the form $\initial \rightarrow \K$ and $\final \rightarrow \K$ respectively. A transition with weight $\0$ is an absent transition, and we enforce that between every two states there is at most one transition per label.
\end{definition}

The definition of deterministic weighted finite-state automaton is fully equivalent to \cref{def:deterministic}, except that now the transitions have weight over  $\K$ instead of over Boolean semifield $\Boolean$. Similarly we can extend the notion of the transition monoid to a \defn{weighted transition monoid} $\wtransmonoid{\wfsa}$, which is generated by weighted transition matrices of $\wfsa$ from $\K^{\states \times \states}$.

\citet{Mohri-on-the-fly} presents a weighted on-the-fly algorithm (see \cref{app:Mohri}) for zero-sum-free semifields which has a very similar paradigm to \cref{alg:onthefly} with the exception that every power state $\powerstate$ is a vector in $\K^\states$. The weights associated with individual states are referred to as \defn{residual weights}. We notice, in the weighted case termination is not always guaranteed.

\begin{lemma}\label{lem:weighted}
    The set of states $\detstates$ of automaton $\detautomaton$ computed by weighted on-the-fly algorithm is a normalized product of $\initf(\initial)$ under the transition monoid $\wtransmonoid{\automaton}$ i. e. $\detstates = $ 
    \begin{equation}
         \left\{ \powerstate \in \K^\states \ \left| \ \powerstate = \frac{ \left({\initf(\initial)}^\top M \right)}{\bigoplus_{i = 1}^{\nstates} {\left({\initf(\initial)}^\top M \right)}_i} \right. \right\} \cup \left\{ \powerstate_\initial \right\} 
    \end{equation}

    \noindent where $M \in \wtransmonoid{\wfsa} \setminus \{ \sI \}$ and $\powerstate_\initial = \initf(\initial)$. 
\end{lemma}

\begin{proof}
    We consider an arbitrary word $\str \in \kleene{\alphabet}$ and show inductively that Mohri's algorithm (see \cref{app:Mohri}) satisfies the property for the power state $\powerstate$ reached by the (unique) path with yield $\str$. Notice, the residual weight of initial state $\powerstate_\initial$ is unchanged and set to $\initf(\initial) \in \K^\states$.
\end{proof}

From \cref{lem:weighted} we can directly derive the extension of \cref{thm:monoid} for weighted automata.

\begin{theorem}\label{thm:weighted-monoid}
    Given a WFSA $\wfsa$ the equivalent deterministic automaton $\detautomaton$ computed by means of weighted on-the-fly will satisfy
    \begin{equation}
        |\detstates| \leq |\wtransmonoid{\wfsa}| + 1
    \end{equation}
    and its construction will require at most $\bigO{\nsymbols \nstates |\wtransmonoid{\wfsa}|}$ steps.
\end{theorem}

\begin{proof}
    From \cref{lem:weighted} we can conclude that every power state $\powerstate \in \detstates$ corresponds to at least one matrix from $\wtransmonoid{\wfsa}$. The bound on the execution time follows from the observation that each power state is processed exactly once.
\end{proof}

As a consequence, if $\wtransmonoid{\wfsa}$ is finite then weighted on-the-fly is guaranteed to terminate. Moreover, \cref{thm:weighted-monoid} allows us to formulate state complexity bounds for weighted automata by only considering monoids of matrices, which forms an interesting open problem.

\section{Conclusion}\label{sec:conclusion}

In this work we have presented a new abstraction for analyzing the performance of the on-the-fly determinization algorithm by establishing a connection with algebraic automata theory. This allowed us to derive state complexity bounds and estimate the asymptotic runtime for various classes of finite-state automata. The summary of our results can be found in \cref{table:summary}.
\begin{algorithm*}[t]
\caption{\textsc{Weighted on-the-fly}}\label{alg:weighted-onthefly}
\begin{algorithmic}[t]
\Ensure $\automaton$ a WFSA with no $\eps$-transitions
\State $\detautomaton \gets (\alphabet, \detstates, \dettrans, \detinitf, \detfinalf)$
\State $\mathsf{stack} \gets \powerstate_\initial$
\Comment{$\powerstate_{\initial}$ is a power set of prior initial states and their weights}
\State $\detstates \gets \{ \powerstate_\initial \}$
\While{$|\mathsf{stack}| > 0$}
    \State \textbf{pop} $\powerstate$ from the $\mathsf{stack}$
    \ForAll{$\syma \in \alphabet $}
        \Comment{Each label corresponds to a new transition}
        \State $\sW_a(\stateq') \gets \bigoplus_{q \in \states} w \otimes r_{\powerstate}(q)  \text{ for all } (\stateq, a, w, \stateq') \in \trans$
        \State $w_a \gets \bigoplus_{q' \in Q} \sW_a(\stateq')$
        \State $\powerstate' \gets \varnothing$
        \ForAll{$\stateq'$ with $(\stateq, a, \bullet, \stateq')$}
            \Comment{Compute residual weights for all power state entries}
            \State $w_{\stateq'} \gets \sW_a(\stateq) \otimes w_a^{-1}$
            \State $\powerstate' \gets \powerstate' \cup \{ (\stateq', w_{\stateq'}) \}$
        \EndFor
        \State $\dettrans \gets \dettrans \cup \{\edge{\powerstate}{\syma}{w_a}{\powerstate'} \}$
        \Comment{Create a transition in a new power state}
        \If{$\powerstate' \notin \detstates$}
            \State $\detstates \gets \detstates \cup \{ \powerstate' \}$
            \State \textbf{push} $\powerstate'$ to the $\mathsf{stack}$
        \EndIf
    \EndFor
\EndWhile
\State $\detfinalf(\powerstate) \gets \bigoplus_{\stateq \in \powerstate} r_{\powerstate}(\stateq) \otimes \finalf(\stateq) \quad \forall \powerstate \in \detstates$
\Comment{Calculate final weights from residual weights}
\State \textbf{return} $\gets \detautomaton$
\end{algorithmic}
\end{algorithm*}

\begin{table}[htb]
    \centering
    \begin{tabular}{lll}
    \toprule
    Constraints & $\boldsymbol{\#}$ states & Runtime  \\
    \midrule
    one-letter & $n^2 - n + 2$ & $\bigO{n^3}$ \\ 
    commutative & $n^{2 \nsymbols}$ & $\bigO{n^{2 \nsymbols + 1}}$ \\ 
    $r$-indecomp. & $2n^{k \log{\nsymbols}}$ & $\bigO{n^{k \log{\nsymbols} + 1}}$ \\
    $\frac{n^2}{k}$-dense$^\dagger$ & $\bigO{n^{k \log{\nsymbols}}}$ & $\bigO{n^{k \log{\nsymbols} + 1}}$ \\ 
    $k$-tree width & $\frac{n^k}{(k-1)!} + 1$ & $\bigO{n^{k+1}}$ \\
    \bottomrule
    \end{tabular}
    \caption{Summarizing state complexity and runtime bounds derived for \cref{alg:onthefly} with $\nstates = n$ over a fixed alphabet $\alphabet$. We assume irreducibility for one-letter and commutative automata and $r$-indecomposability with $r \geq \lceil n / (k \log{n})\rceil$. $\dagger$ indicates that the bound holds with high probability.}
    \label{table:summary}
\end{table}

The bounds demonstrate that there are many different properties of automata which can limit the underlying non-determinism and that these properties are intrinsically non-uniform, and seem to have little in common. This leads us to the idea that algebraic properties of automata might be the method to quantify and measure non-determinism and unify many results present in the field. Although we didn't discuss it explicitly in the paper, some of the criteria --- such as the strongly connected property --- can hold just partially and still translate into limited state complexity (for example if an automaton is a concatenation of two strongly connected automata a similar argument can be made). Thus a straightforward extension of our approach might lead to more complexity bounds.

On the other hand, the algebraic method has some limitations that yet have to be resolved, most importantly the missing connection to the starting and final states. Closing (or at least measuring) this ``complexity gap'' remains an important open question. 

Our methodology can be extended and applied to weighted automata over general semirings. Analyzing matrix monoids with different semirings might lead to new insights and lay the foundation for the complexity theory of weighted finite-state automata, something that to our knowledge has not been formalized until this day.

\appendix

\section{Mohri's algorithm for weighted determinization} \label{app:Mohri}

Here we present a version of the weighted on-the-fly algorithm by \citet{Mohri-on-the-fly} over general semifields $\langle \K, \oplus, \otimes, \0, \1 \rangle $. In contrast to the original version we additionally assume commutativity under $\otimes$ for \cref{lem:weighted} to be applicable. The pseudocode for the full algorithm can be found in \cref{alg:weighted-onthefly}, and was adapted slightly to resemble the non-weighted determinization, as presented in \cref{alg:onthefly}. However there are some important differences that we will clarify.

We consider the output automaton $\detautomaton$ from the \cref{alg:weighted-onthefly}. Its state space $\detstates$ is a subset of $\K^\states$, meaning that every state $\powerstate \in \detstates$ is power state and additionally each state withing $\powerstate$ has a weight from $\K$ associated with it. However this weight is in no way connected to the initial weight function $\detinitf$ or final weight function $\detfinalf$. To avoid ambiguity we will refer to it as \defn{residual weight} and use $r_\powerstate$ function to refer to the residual weights for some power state $\powerstate$.

Further we observe that the initial weight function $\detinitf$ is essentially Boolean, since it only has weight $\1$ for initial power state $\powerstate_\initial \coloneqq \{ (\stateq, \initf(\stateq)) \ | \  \stateq \in \states, \ \initf(\stateq) \neq \0 \}$ and $\0$ everywhere else. The final weight function $\detfinalf(\powerstate)$ for some power state $\powerstate \in \detstates$ on the other hand is simply the dot product of all residual weights of $\powerstate$ with their non-deterministic final weights $\finalf$. Thus the main component carrying all information are the residual weights $r_{\powerstate}$.

The execution paradigm is very similar to \cref{alg:onthefly} where we start in the power state $\powerstate$ (initially $\powerstate_\initial$) and then ``explore'' all possible states and push them on the stack. Moreover, now we add the weights to individual transitions by summing up the weights from all transitions leaving the power state (see the calculation of $w_a$). The newly created power state will have the residual weights calculated very similarly, but the summation goes over individual weights (see the calculation of $\sW_a(\stateq'))$. Most importantly before assigning the residual weights to the new power state $\powerstate'$ we normalize them by the transition weight $w_a$ to keep the weight of the path unchanged.

\bibliography{tacl2021}
\bibliographystyle{acl_natbib}

\end{document}

%% file: macros.tex
\usepackage{float}
\usepackage{graphicx}
\usepackage{amssymb}
\usepackage{amsmath}
\usepackage{amsthm}
\usepackage{mathtools}
\usepackage{svg}
\usepackage{algorithm}
\usepackage[noend]{algpseudocode}
\usepackage{booktabs}

\usepackage{tikz}
\usetikzlibrary{automata, fit, positioning, matrix, arrows.meta, decorations.pathmorphing, shapes, shapes.multipart, chains,shadows.blur}

\usepackage{tikz-qtree}

\usepackage{tikz-dependency}

\tikzset{
-{Latex[length=1.8mm, width=1.4mm]}, 
every initial by arrow/.append style={anchor/.append style={shape=coordinate}},
every state/.style={thick, fill=gray!10},
every loop/.append style=-{Latex[length=1.8mm, width=1.4mm]},
initial text=,
initial distance=0.5cm,
line width=0.7pt,
node distance=15mm,
minimum size=4mm,
}

\newcommand{\mymacro}[1]{{#1}}

\newcommand{\defn}[1]{\textbf{#1}}

\usepackage{colortbl}

\definecolor{ETHBlue}{RGB}{33,92,175}	
\definecolor{ETHGreen}{RGB}{98,115,19}		
\definecolor{ETHPurpleDark}{RGB}{140,10,89}	
\definecolor{ETHPurple}{RGB}{163,7,116}	
\definecolor{ETHGray}{RGB}{111,111,111}	
\definecolor{ETHRed}{RGB}{183,53,45}	
\definecolor{ETHPetrol}{RGB}{0,120,148}	
\definecolor{ETHBronze}{RGB}{142,103,19}	

\colorlet{ETHdarkblue}{ETHBlue!80!black}
\colorlet{ETHdarkgreen}{ETHGreen!80!black}
\colorlet{ETHpink}{ETHPurple}
\colorlet{ETHgray}{ETHGray}
\colorlet{ETHred}{ETHRed}
\colorlet{ETHgreenblue}{ETHPetrol}
\colorlet{ETHbrown}{ETHBronze}

\colorlet{MacroColor}{ETHPetrol}
\colorlet{MACROCOLOR}{MacroColor}

\usepackage{cleveref}
\crefname{section}{\S}{\S\S}
\crefname{table}{Tab.}{Tabs.}
\crefname{figure}{Fig.}{Figs.}
\crefname{algorithm}{Algorithm}{Algorithms}
\crefname{equation}{Eq.}{Eqs.}
\crefname{example}{Example}{Examples}
\crefname{fact}{Fact}{Facts}
\crefname{appendix}{Appendix}{Appendices}
\crefname{theorem}{Theorem}{Theorems}
\crefname{reTheorem}{Theorem}{Theorems}
\crefname{aquestion}{Question}{Questions}
\crefname{assumption}{Assumption}{Assumptions}
\crefname{lemma}{Lemma}{Lemmas}
\crefname{reLemma}{Lemma}{Lemmas}
\crefname{proposition}{Proposition}{Propositions}
\crefname{chapter}{Chapter}{Chapters}
\crefname{line}{line}{lines}
\crefname{principle}{Principle}{Principles}
\crefname{definition}{Definition}{Definitions}
\crefname{corollary}{Corollary}{Corollaries}
\crefname{Exercise}{Exercise}{Exercises}
\crefname{observation}{Observation}{Observations}
\crefformat{section}{\S#2#1#3} 

\newtheorem{definition}{Definition}[section]
\newtheorem{theorem}{Theorem}[section]
\newtheorem{observation}[theorem]{Observation}

\newtheorem{corollary}[theorem]{Corollary}
\newtheorem*{remark}{Remark}

\newtheorem{proposition}[theorem]{Proposition}

\newtheorem{lemma}[theorem]{Lemma}

\usepackage{algorithm}
\usepackage[noend]{algpseudocode}  
\algrenewcommand\algorithmicindent{1.0em}%
\newcommand{\rightcomment}[1]{{\color{gray} \(\triangleright\){\footnotesize\textit{#1}}}}
\algrenewcommand{\algorithmiccomment}[1]{\hfill \rightcomment{#1}}  
\algnewcommand{\LinesComment}[1]{\State {\color{black!50!green}\rightcomment{\parbox[t]{.95\linewidth-\leftmargin-\widthof{\(\triangleright\) }}{#1}}}}
\algnewcommand{\LineComment}[1]{\State {\color{black!50!green}\smaller \(\triangleright\) \parbox[t]{\linewidth-\leftmargin-\widthof{\(\triangleright\) }}{\it #1}\smallskip}} 
\algnewcommand{\InlineComment}[1]{\hfill {\color{black!50!green}\(\triangleright\) {\scriptsize \it #1}}}
\algrenewcommand\algorithmicindent{1.0em}%
\algrenewcommand\alglinenumber[1]{{\tiny\color{black!50}#1.}\hspace{-2pt}}
\newcommand{\algorithmicfunc}[1]{\textbf{def} {#1}:}
\algdef{SE}[FUNC]{Func}{EndFunc}[1]{\algorithmicfunc{#1}}{}
\makeatletter
\ifthenelse{\equal{\ALG@noend}{t}}%
  {\algtext*{EndFunc}}
  {}%
\makeatother

\newcommand{\R}{{\mymacro{ \mathbb{R}}}}
\newcommand{\K}{{\mymacro{ \mathbb{K}}}}

\newcommand{\monoid}{{\mymacro{\mathbb{M}}}}
\newcommand{\monoidS}{{\mymacro{\mathbb{S}}}}
\newcommand{\transmonoid}[1]{{\mymacro{\mathbb{S}(#1)}}}
\newcommand{\wtransmonoid}[1]{{\mymacro{\mathbb{W}(#1)}}}
\newcommand{\morphismdef}{{\mymacro{\mu}}}
\newcommand{\morphism}[1]{{\mymacro{\mu(#1)}}}
\newcommand{\permgroup}{{\mymacro{S_n}}}

\newcommand{\indexM}[1]{{\mymacro{\textup{index}(#1)}}}
\newcommand{\periodM}[1]{{\mymacro{\textup{period}(#1)}}}
\newcommand{\0}{{\mymacro{\mathbf{0}}}}
\newcommand{\Boolean}{{\mymacro{\mathbb{B}}}}
\newcommand{\Tropical}{{\mymacro{\mathbb{T}}}}
\newcommand{\tw}[1]{{\mymacro{\textsc{tw}(#1)}}}

\newcommand{\alphabet}{{\mymacro{ \Sigma}}}

\newcommand{\kleene}[1]{{\mymacro{#1^*}}}

\def\1{\mathbf{1}}

\def\eps{{\mymacro{ \varepsilon}}}

\newcommand{\str}{{\mymacro{\boldsymbol{y}}}}

\newcommand{\syma}{{\mymacro{a}}}

\newcommand{\set}[1]{{\mymacro{\left\{ #1 \right\}}}}
\newcommand{\powerset}[1]{{\mymacro{\mathcal{P}{\left(#1\right)}}}}

\newcommand{\N}{{\mymacro{ \mathbb{N}}}}

\newcommand{\nstates}{{\mymacro{ |\states|}}}
\newcommand{\detstates}{{\mymacro{ \states_{\textsc{det}}}}}
\newcommand{\nsymbols}{{\mymacro{ |\alphabet|}}}

\newcommand{\graph}{\mymacro{ \mathcal{G}}}

\newcommand{\automaton}{{\mymacro{ \mathcal{A}}}}
\newcommand{\wfsa}{{\mymacro{ \automaton}}}

\newcommand{\detautomaton}{{\mymacro{ \automaton_{\textsc{det}}}}}

\newcommand{\stateq}{{\mymacro{ q}}}

\newcommand{\states}{{\mymacro{ Q}}}

\newcommand{\trans}{{\mymacro{ \delta}}}
\newcommand{\dettrans}{{\mymacro{ \delta_{\textsc{det}}}}}

\newcommand{\apath}{{\mymacro{ \boldsymbol \pi}}}

\newcommand{\initial}{{\mymacro{ I}}}

\newcommand{\final}{{\mymacro{ F}}}
\newcommand{\detfinal}{{\mymacro{ F_{\textsc{det}}}}}
\newcommand{\initf}{{\mymacro{ \lambda}}}
\newcommand{\detinitf}{{\mymacro{ \lambda_{\textsc{det}}}}}
\newcommand{\finalf}{{\mymacro{ \rho}}}
\newcommand{\detfinalf}{{\mymacro{ \rho_{\textsc{det}}}}}

\newcommand{\languagefsa}[1]{{\mymacro{\mathcal{L} ( #1 )}}}

\newcommand{\fsatuple}{{\mymacro{ \left( \alphabet, \states, \initial, \final, \trans \right)}}}
\newcommand{\wfsatuple}{{\mymacro{ \left( \alphabet, \states, \trans, \initf, \finalf \right)}}}

\newcommand{\edgenoweight}[3]{#1 \xrightarrow{#2} #3}
\newcommand{\edge}[4]{{\mymacro{#1 \xrightarrow{#2 / #3} #4}}}

\newcommand{\powerstate}{{\mymacro{ \mathcal{Q}}}}
\newcommand{\initpowerstate}{{\mymacro{ \mathcal{Q}_{\initial}}}}

\def\sI{{{\mymacro{ \mathcal{I}}}}}

\def\sL{{{\mymacro{ \mathcal{L}}}}}

\def\sT{{{\mymacro{ \mathcal{T}}}}}

\def\sW{{{\mymacro{ \mathcal{W}}}}}

\newcommand{\bigO}[1]{{\mymacro{ \mathcal{O}\left(#1\right)}}}
\newcommand{\bigOmega}[1]{{\mymacro{ \Omega\left(#1\right)}}}

